\documentclass[10pt]{sig-alternate-10pt}

\usepackage{epsf}
\usepackage{epsfig}
\usepackage{mathrsfs}
\usepackage{subfigure}
\usepackage{graphics}
\usepackage{graphicx}
\usepackage{epstopdf}

\newcommand{\comment}[1]{}
\renewcommand{\textfloatsep}{3pt}

\begin{document}

\title{A Simple Conceptual Generator for the Internet Graph}

\numberofauthors{3}
\author{
	\alignauthor Theodoros Lappas \\
		\affaddr{Dept. of CS\&E} \\ 
		\affaddr{UC Riverside} 
		\email{tlappas@cs.ucr.edu}
	\alignauthor Konstantinos Pelechrinis \\
		\affaddr{Dept. of CS\&E} \\
		\affaddr{UC Riverside} \\
		\email{kpele@cs.ucr.edu}
	\alignauthor Michalis Faloutsos \\
			\affaddr{Dept. of CS\&E} \\ 
		\affaddr{UC Riverside} 
		\email{michalis@cs.ucr.edu}
}

\date{}
\maketitle 
\thispagestyle{empty}
\pagestyle{empty}

\begin{abstract}

The evolution of the Internet during the last years, has lead to a dramatic increase of the size of its graph at the Autonomous System (AS) level.  
Soon - if not already - its size will make the latter impractical for use from the research community, e.g. for protocol testing (AS level routing protocols).  
Reproducing a smaller size, snapshot of the AS graph is thus important.  
However, the first step towards this direction is to obtain the ability to faithfully reproduce the full AS topology.   
The objective of our work, is to create a generator able to accurately emulate and reproduce the distinctive properties of the Internet graph. 
Our approach is based on 
\textbf{(a)} the identification of  the jellyfish-like structure \cite{Siganos06} of the Internet and 
\textbf{(b)} the consideration of the peer-to-peer and customer-provider relations between ASs. 
We are the first to exploit the distinctive structure of the Internet graph together with utilizing the information provided by the AS relationships in order to create a {\em tool} with the aforementioned capabilities.   
Comparing our generator with the existing ones in the literature, the main difference is found on the fact that our tool does not try to satisfy specific metrics, but tries to remain faithful to the conceptual model of the Internet structure. 
In addition, our approach can lead to 
\textbf{(i)} the identification of important attributes and patterns in the Internet AS topology, as well as,
\textbf{(ii)} the extraction of valuable information on the various relationships between ASs and their effect on the formulation of the Internet structure.  
We implement our graph generator and we evaluate it using the largest and most recent available dataset for the AS topology.  
Our evaluations, clearly show the ability of our tool to capture the structural properties of the Internet topology at the AS level with high accuracy.   
Finally, we discuss the potentials of our generator not only to reproduce, but also to shrink the input graph while maintaining its unique structure and properties.
\end{abstract}


\section{Introduction}
\label{sec:intro}
\setcounter{paragraph}{0}

The growing importance, usage and size of the Internet has increased the interest for the Internet topology among the research community; 
numerous research ventures have focused on the study of the Intenet graph. 
A thorough understanding of the Internet topology at the AS level can contribute significantly to the enhancement of fundamental applications such as routing and information diffusion.  
In particular, accurate topology information is necessary for numerous reasons; 
\textbf{(a)} {\em simulating} real networks requires the network topology to have been firstly obtained, 
\textbf{(b)} {\em network management} descisions are augmented by topology knowledge (e.g. finding network bottlenecks, deciding about the placement of new routers etc), while 
\textbf{(c)} {\em topology aware protocols} can significantly improve the overall performance enjoyed from the end users (e.g. QoS routing etc).

As part of the general effort to effectively describe the Internet, various graph metrics have been introduced that aim to help towards this goal \cite{Mahadevan06}.  
Applying graph mining techniques and evaluating corresponding metrics to capture and describe the various properties of the Internet can provide us with valuable information for reproducing the topology, as well as evaluating the {\em similarity} between the original and the generated graph. 
An accurate generator would have numerous applications including - but not limited to - the generation of artificial graphs for simulations when real data is not available, the examination of  {\em "what-if"} scenarios, extrapolations on the evolution of the graph structure and also the creation of representative miniature graphs in order to reduce the computational cost of simulations and graph based procedures. 
Furthermore, the generation process can provide us with useful insights on the graph evolution over time and the latent factors that affect and drive its evolution \cite{Leskovec05} helping on descisions for network related operations as mentioned above.

Our objective in this work, is to design and implement an accurate graph generation process that can reproduce the real Internet AS level topology.   
Our main contributions can be summarized in the following:

\begin{itemize}
\item We identify and study the jellyfish structure of the AS level Internet topology.  
For the purposes of our work we use the largest and most detailed snapshot of the Internet available in the literature \cite{He}.
\item We design a graph generation process based on the conceptual jellyfish model and the AS relations.
\item We implement and thoroughly evaluate our generator using various graph metrics that can capture the structural properties of a graph. 
\end{itemize}

\textbf{Scope of our work: }
There are many graph generators proposed in the literature thus far\footnote{A more detailed discussion on these works is given in Section \ref{sec:RelatedWork}.}.  
Nevertheless, our work follows a different approach from all the existing ones.  
In particular, the novelty of our generative scheme, is mainly related with the exploitation of the conceptual model of the real graph.  
Despite the fact that we are mainly focused on the Internet AS graph (and its jellyfish - like conceptual model \cite{Siganos06}), we would like to stress out that {\bf our general point of view can be applied to any other graph generator.  
A simple, conceptual model is only required, that represents the original graph to be reproduced}.

The rest of the paper is organized as follows.  
In the next section we briefly give a background on graph theory and refer to related work on graph generators.  
In section \ref{sec:TheJellyfishModel} we present the jellyfish model which forms the basis for our approach.  
Section \ref{sec:RealGraph} presents our measurements on the real Internet graph at the AS level.  
In section \ref{sec:GraphGenerator} we describe our graph generator, while in section \ref{sec:Evaluation} we present experimental results for the performance of our approach.  
Finally, we conclude with a brief overview and ideas for future work in Section \ref{sec:conclusions}. 
\section{Backgroung and Related Studies}
\label{sec:RelatedWork}
\setcounter{paragraph}{0}

In this section we provide a brief background on graph theory concepts and related work existing in the literatrure.  

\subsection{Graph Metrics}
\label{sec:metrics}

Every network can be represented through a graph $G=(V,E)$.  
The set $V$ of the vertices, can represent various network entities, depending on the abstraction level for the graph.  
For example, each vertex might represent a terminal station, a router or even an AS.  
For the purposes of our study each vertex represents an AS.  
The set $E$ of the edges , represents the links between the nodes.  
The graph can be directed or undirected.  
For the case of the AS graph, the direction of an edge can potentially represent business relations between the different ASs (customer-provider, p2p).
Metrics used in order to describe a graph include: 

\textbf{Node Degree:} 
The degree of a node, is the number of edges incident to the corresponding vertex representing the node.  
For node $v$, $deg(v)$ is used for refering to the node degree.  
The maximum degree of a graph, $\Delta(G)$, is:

\begin{equation}
\Delta(G) = \max_{k \in V}\{deg(k)\}
\label{eq:max_d}
\end{equation}

while the minimum degree of a graph, $\delta(G)$, is given by:
 
\begin{equation}
\delta(G) = \min_{k \in V}\{deg(k)\}
\label{eq:min_d}
\end{equation}

Finally, the average node degree of a graph, $\overline{D(G)}$, is: 

\begin{equation}
\overline{D(G)} = \dfrac{\displaystyle\sum_{i=1}^{\vert V \vert} deg(i)}{\vert V \vert} 
\label{eq:avg_d}
\end{equation}

For directed graphs, we can further define the {\bf indegree} and {\bf outdegree} of the node, depending on whether the edge is ending at or originating from the node respectively.  

\textbf{Node Degree Distribution $P(k)$: }
The degree distribution, is the probability distribution of the node degree over the whole network/graph.  
If we assume that there are $n_k$ nodes with degree $k$, then the degree distribution, $P(k)$, is simply: 

\begin{equation}
P(k)=\dfrac{n_k}{|V|}
\label{eq:degree_distr}
\end{equation}

In a similar way, for directed graphs, we can define the  {\bf indegree} and {\bf outdegree} distributions.  

\textbf{Diameter : }
Let's denote with $d(u,v)$, the distance between nodes $u$ and $v$, i.e. the number of edges in the {\bf shortest path} connecting vertices $u$ and $v$.  
Then, the diameter D(G) of the graph G, is given from: 

\begin{equation}
D(G) = \max_{u,v\in V}\{d(u,v)\}
\label{eq:diameter}
\end{equation}

From the above definition it is evident that for every connected pair $(u,v)$ of nodes, we have that: $d(u,v) < D(G)$.  
Relaxing the above inequality, and having it satisfied from $90\%$ of all connected pairs of nodes of graph G, we get the definition of the {\em effective} diameter.  
Formally, the {\em effective} diameter, $D_{eff}(G)$, is defined as the maximum distance in which $90\%$ of all connected pairs of nodes can reach each other.  
It is obvious that $D_{eff}(G) \le D(G)$.

The above metrics have been extensively used in the literature, in order to capture and describe various graph properties.

\subsection{Static and Temporal Graph Generators}
\label{sec:approach}

To this end, a significant number of graph generators have been proposed. 
Some of these attempts, including our own, are focused on the Internet topology \cite{Fabrikant02} \cite{Chang06} while others view the generation process as a general graph problem \cite{Leskovec05}. 
Based on the generative approach followed by the various generators we can further separate them into two groups:

\begin{itemize}

\item Generators that receive as input a static snapshot of a graph and attempt to create a duplicate with identical properties. 
The most common static patterns that these graph generators try to match are the \textit{degree distribution} 
and the \textit{diameter} of the graph. 
A lot of work has been done on the study of static patterns for various types of graphs.  
As some examples, the authors in \cite{Faloutsos3} \cite{Kleinberg99} \cite{Broder00} \cite{Redner98} \cite{Chakrabarti04} study the degree distribution of a large variety of graphs that span a huge spectrum, from the Internet and Web graphs to citation and online social networks' graphs.  Moreover, \cite{Albert02} \cite{Milgram67} take a close look on the diameter of the Internet and Web graphs as well as that of social networks' graphs. 
\item Generators that study the evolution of a given graph and attempt to emulate it. 
A lot of work has been done on the temporal evolution of graphs and the underlying laws that dictate this evolution.  
The work in \cite{Leskovec05} is a representative example of an evolution-based generator that considers temporal properties such as the densification power law and the shrinking diameter. 

\end{itemize}

Generators based on the temporal evolution, try to track and emulate the gradually development of a graph, which constitutes a rather upredictable and unstable process.  
Clearly, this can have a negative effect on the accuracy of the generation process.  
Moreover, focusing especically on the Internet graph, the lack of information on the past states and the evolution of the Internet topology, forms another significant obstacle towards adopting an evoluation-based AS level generator.  
Our approach belongs to the first family of generators, and is thus unaffected by these problems.

A process closely related to graph generation, is that of {\bf graph shrinking}.  
The challenge here is to create a smaller graph while maintaining the important properties and dominant structure of the original. 
As a result, simulations that were too expensive to run on the large original graph can instead be applied to the miniature representation. 
Clearly, generators belonging in the first category - like ours - can be potentially used for graph shrinking.  
On the other hand, evolution-based approaches, will create smaller snapshots with properties different than the real current topology.  
Capturing evolution properties, these generators will reproduce the graph (smaller in size) as it was in the past, when having this smaller size.  
Clearly this is not what we seek for, since many graph properties change with time, as alluded above.

\subsection{Graph Generating Models}
\label{sec:models}

The simplest model for a graph generator is the \textit{Random Graph Model}, proposed by Erdos and Renyi \cite{Erdos60}; 
introduced in the early '60s it constitutes the first graph generator that was ever presented.  
With this approach, each pair of nodes has the same, independent probability of becoming connected through an edge.  
In other words, starting from a set of nodes we add random edges between them, with each edge having the same, independent probability. 
There is clearly a tradeoff between simplicity and accuracy in the Random Graph Model, since the above procedure is not able to generate graphs that match the properties observed at most of the real life ones.

Many recent models for graph generators are based on \textit{preferential attachement} \cite{Albert99} \cite{Albert02} \cite{Kleinberg99} \cite{Kumar99} \cite{Winick02}.  
The intuition behind this approach, is that nodes with high degrees will attract even more nodes, thus, {\bf the rich gets richer}. 
Simply put, in these models, the new nodes that are added to the graph at each repetition of the algorithm, prefer to connect to nodes with a high degree. 
This approach can create graphs with degree distributions similar to the ones noticed in real graphs\footnote{This is clearly related with the fact that real world graphs follow {\bf power laws} \cite{Faloutsos3}; there are a few vertices with high node degree and many vertices with low node degree.}. 
Finally, graphs generated with preferential attachment, tend to exhibit slowly increasing diameters with the cardinality of set $V$.

Briefly, other approaches that have been proposed for graph generation include, the \textit{small-world generator} \cite{Watts98} and the \textit{Waxman generator} \cite{Waxman88}.  
Finally, Fabrikant \textit{et al} \cite{Fabrikant02} have studied the general problem of creating a graph when resource constraints are existent.  

Most of the above proposed schemes try to match a limited set of properties - or even just a single one - and thus fail to capture the general conceptual model behind the graph structure.  	
{\em To the best of our knowledge, we are the first that try to exploit the conceptual model of the original topology (in our case that of the Internet AS graph) towards implementing an accurate graph generator.}

 \section{The Jellyfish model}
\label{sec:TheJellyfishModel}
\setcounter{paragraph}{0}

In this section we will present the conceptual model for the Internet topology that we considered.

Understanding the conceptual underlying structure of complex networks is important towards creating a generation tool.   
Measuring and acquiring metric values (e.g. degree distribution, diameter etc) can formally describe a network graph.  
Nevertheless, as a further step, it is important to attain a model that can capture more {\em visual} information for the network.  

Siganos \textit{et al} \cite{Siganos06} have proposed a conceptual model for the Internet graph at the AS level. 
The authors managed to create an effective conceptual model of the Internet, which could be easily understood and depicted.   
Furthermore, they showed that their model captures the most signicant topological properties of the inter-domain level of the Internet. 
In the following paragraphs we will try to provide a brief, yet complete, presentation of the jellyfish model. 
 
The jellyfish model defines a hierarchy of the graph nodes (ASs). 
First, the {\bf core} of the graph is identified. 
This core is essentially a clique of high degree nodes, all connected to each other through {\em peer-to-peer (P2P)} links. 
Once the core is identified, the rest of the nodes can easily be classified. 
Thinking of the core as the head of the jellyfish, the rest of the nodes are distributed in {\bf shells} that surround the head. 
The first shell contains all the nodes that are adjacent to the core, except the one degree nodes. 
Recursively, every shell contains all the nodes (except the one degree nodes) that are connected through an edge to some node that belongs to the previous shell.  
The one degree nodes are represented as {\bf hangers}, hanging from the shell of the other end of their single edge. 
A visualization of this model can be found in 
\cite{Siganos06}.


\begin{table*}[ht]

	\begin{center}
		\centerline{\begin{tabular}{|l|l|l|l|l|}
			\hline
			Ring & \# of nodes & Percentage of nodes & \# of P2P edges & \# of Customer-Provider edges \\ \hline
			$0$ & $9$ & $0.05\%$ & $36$ & $0$ \\ \hline
			$1$ & $6419$ & $32.19\%$ & $12873$ & $7396$ \\ \hline
			$2$ & $6102$ & $30.6\%$ & $1167$ & $1481$ \\ \hline
			$3$ & $1245$ & $6.24\%$ & $165$ & $190$ \\ \hline
			$4$ & $151$ & $0.76\%$ & $3$ & $8$ \\ \hline
			$5$ & $6$ & $0.03\%$ & $0$ & $0$\\
			\hline
		\end{tabular}}
	\end{center}
\caption{The distribution of the nodes and types of edges among the various rings of the jellyfish model.}
\label{tab:distro}
\end{table*}

By using this conceptual model, the authors in \cite{Siganos06} identified six layers in the Internet topology, counting the core as layer 0.  
Their observations can be summarized in the following: 

\begin{itemize}

\item Approximately 80-90\% of nodes are in the first 3 layers.
\item Network grows "horizontally", which means that the evolution of the graph doesn't result in more layers but in larger-denser layers.   
\item The topological importance of a shell decreases as we move away from the core.  
\item Most of the connectivity is towards the center.  This means that nodes in outer shells needs to route through previous shells for most of their shortest path connections.
\item The nodes in the first three layers are within 5 hops from each other.

\end{itemize}
 
This simple and easy even to visualize model, captures many important properties of the AS Internet topology.  
In particular, the jellyfish model, can accurately express the following facts/characteristics of the Internet AS graph:

\begin{itemize}

\item The network is compact, as 99\% of the nodes are within 6 hops; 
the jellyfish model exhibits the same diameter \cite{Siganos06}.
\item There exists a highly connected center in the Internet.  
This center corresponds to the clique of the high degree nodes, defined at the jellyfish model as the core.
\item There exists a loose hierarchy at the internet; 
nodes far from center are less important at the jellyfish model.
\item One-degree nodes are scattered everywhere; 
hangers hang from all jellyfish's shells.
\item The network has the tendency to be one large connected component (core).  

\end{itemize} 

Among the other contributions of this model, is that it forms a benchmark that can be used to evaluate the performance of various graph generators. 
The degree to which a generated graph demonstrates the above structure can be a testament of how well the graph generation tool used can capture the actual Internet graph characteristics.

In this work, we focus on the implementation of a generator for the Internet AS topology, able to produce graphs that are faithful to the jellyfish model.
For the purposes of our study we use the most accurate snapshot of the AS topology existing in the literature \cite{He}.  
We start off by identifing the jellyfish structure of this snapshot and its specific properties in the next section. 
We then study the extracted structure and use the obtained information to generate the new graph (Section \ref{sec:GraphGenerator}).
\section{The real internet graph}
\label{sec:RealGraph}
\setcounter{paragraph}{0}

In this section we present the data and the information obtained from the real Internet topology.

In order to drive the implementation and test the accuracy of our graph generator we need a detailed snapshot of the real Internet topology.  
The latter will help us (i) to {\em extract} the jellyfish-related parameters of our model (e.g. number of shells, size of shells, nodes distributions etc) and (ii) to {\em compare} the topology created from our generator with the actual one.  
For the needs of our work, we used the snapshot collected from Yihua He at UCR \cite{He}.  
To this point this is the largest available snapshot of the internet.  
The graph obtained is part of a larger project in which the missing AS links from the commonly-used Internet topology snapshots are to be identified \cite{He}.  
This procedure includes cross validation of BGP routing tables, Internet Routing Registries and traceroute data.  
An interesting point made is that most of the missing peer-to-peer AS links are found to be Internet Exchange Points (IXP) links.  
More details for the project and the data used can be found in \cite{He1} \cite{He2} \cite{He3} \cite{He4}.

Using the above data, we extracted valuable information with regards to the conceptual representation of the graph as a jellyfish.   
We observe that approximately 87-90\% of the total number of nodes, belong either to the core (we will also refer to the core as ring 0) or to the first 2 rings (shells), or are hangers originating from rings 0-2. 
Later, during the presentation of our graph generator, it will be more clear why this information is important.  
Additionaly, we would like to emphasize on the fact that only 9 nodes belong to the core of the jellyfish.

\begin{table}[h]

	\begin{center}
		\centerline{\begin{tabular}{|l|l|l|}
			\hline
			Origin ring & \# of nodes & Percentage of nodes\\ \hline
			$0$ & $1254$ & $6.29\%$ \\ \hline
			$1$ & $2912$ & $14.6\%$  \\ \hline
			$2$ & $1420$ & $7.12\%$ \\ \hline
			$3$ & $377$ & $1.89\%$ \\ \hline
			$4$ & $39$ & $0.20\%$ \\ \hline
			$5$ & $1$ & $0.005\%$\\
			\hline
		\end{tabular}}
	\end{center}
\caption{The distribution of the hanger of the jellyfish model.}
\label{tab:hangers}
\end{table}

\begin{table}[h]

	\begin{center}
		\centerline{\begin{tabular}{|l|l|l|}
			\hline
			Bridge & \# of P2P edges & \# of Customer-Provider edges \\ \hline
			$0-1$ & $521$ & $9104$ \\ \hline
			$0-2$ & $91$ & $0$  \\ \hline
			$0-3$ & $2$ & $0$ \\ \hline
			$1-2$ & $5532$ & $11679$ \\ \hline
			$1-3$ & $261$ & $930$ \\ \hline
			$1-4$ & $24$ & $56$\\ \hline
			$1-5$ & $2$ & $0$ \\ \hline
			$2-3$ & $514$ & $1216$ \\ \hline
			$2-4$ & $27$ & $87$\\ \hline
			$3-4$ & $24$ &$106$\\ \hline
			$3-5$ & $1$ & $2$ \\ \hline
			$4-5$ & $0$ & $6$\\
			\hline
		\end{tabular}}
	\end{center}
\caption{The distribution of the type of edges between nodes at different rings.}
\label{tab:type}
\end{table}

Tables \ref{tab:distro} - \ref{tab:type} contain important information for the implementation of our graph generator.  
In particular, these tables include values for the input parameters needed from our model\footnote{More details for the computation of these parameters from the Internet snapshot is given in the following section.}.  
In a nutshell, we present data for: 

\begin{itemize}
\item The percentages of the nodes that belong to a specific ring.
\item The percentages of the nodes that are hangers originating from a specific ring.
\item The percentages of P2P and CP (customer-provider) edges that exist within each ring.
\item The percentages of P2P and CP edges that belong to each {\bf bridge}. 
A bridge XY is defined as the set of edges connecting nodes from ring X to nodes from ring Y.
\end{itemize}

In the following section we elaborate on our generator. 
As we show, \textbf{{\em our model requires a limited number of simple parameters so as to capture and reproduce the conceptual model of the Internet topology}}.
\section{The model of our graph generator}
\label{sec:GraphGenerator}
\setcounter{paragraph}{0}

In this section we will present the steps towards creating our graph generator.

For the purposes of our work, and in order to capture the principles of the jellyfish structure of the Internet topology, we used a simple deterministic approach. 
The hierarchical structure of our model has three distinct components: (i) rings, (ii) bridges and (iii) hangers. 
In the following, a brief explanation is given on how we dealt with each of these components in our generator.

\textbf{Rings:} 
These are essentially the shells of the jellyfish. 
We start by identifying the core (clique), which we call ring 0. 
After that, we populate the remaining rings based on the jellyfish model. 
For each ring, we calculate the number of nodes (as a percentage of the total) and also the percentages of P2P and Customer-Provider edges that exist within the ring.   

\textbf{Bridges:} 
These are the parts that connect the shells of the jellyfish. 
A bridge contains no nodes, but only the edges connecting the nodes of two specific rings. 
For each bridge, we calculate the percentages of P2P and Customer-Provider edges that belong to the bridge.  
It is entirely possible for a bridge between two specific rings to contain zero edges.  

\textbf{Hangers:} 
These are the hangers of the jellyfish model. 
Each ring is paired with a set of hangers. 
For example hanger-set 1 belongs to ring-1 and contains the 1-degree nodes that stem from ring 1. 
For each hanger-set we calculate the number of nodes (as a percentage of the total) belonging into it.

After creating the jellyfish and populating its various parts, based on the real Internet topology, the next step is to study the rings in more detail and start building our own graph. 
After processing the given snapshot, it was clear that the P2P edges within each ring follow a {\em power law}, with a {\em different coefficient for each ring}. 
The coefficient of each ring is calculated and stored.  

The study of CP edges is more complicated, since there are {\bf constraints} that need to be taken into account. 

\begin{enumerate}
\item Loops need to be considered. 
A customer of a specific node cannot be at the same time its provider (or recursively, a provider of its providers).  
\item In order to emphasize to the hierarchical jellyfish structure, some nodes have to be moved up in the ring hierarchy so as to ensure that a node's provider can only exist either in a higher ring or in the same ring as the node. 
\item Connectivity needs to be assured. 
Even though this turns out to be trivial after the application of our strategy and the above mentioned constraints, there could be cases where some extra edges need to be added to ensure that the connectivity of the structure is maintained. 
\end{enumerate} 
With the above constraints in mind, we implement a {\bf rich gets richer} technique that accurately matches the actual graph. 
The important thing to note here is that for each ring of the original graph we need to calculate and store the pace at which the rich gets richer (coefficient). 
Finally, a bias is introduced to increase the number of nodes with exactly 2 providers, which is the common case in the actual snapshot. 
With the exception of nodes with 2 customers, that constitute the vast majority of the nodes, the rich gets richer strategy manages to match the provider distribution without further interference.

 A final constraint (constraint 4) is used to control the effect of our generation technique. 
 In particular, for each ring, an upper bound on the maximum and minimum number of CP and P2P edges for a single node is placed, based on the respective bounds noted in the actual ring. 
 In terms of implementation, this means that nodes with number of edges currently below the minimum are generally preferred, while nodes that reach the upper bound are excluded after they reach this maximum.

{\bf Generation Process: }
Having all the above parameteres and constraints in mind, we start building our graph by creating the core and the corresponding clique using P2P links. 
We then populate the rings by simply adding the respective number of nodes to each ring.  
After the nodes have been added, we add the P2P edges within the rings using the calculated power law coefficients.  
The rich gets richer approach is then applied to the resulting structure in order to populate the bridges with P2P edges. 

The next step is the addition of the CP edges. 
A similar strategy as above is used, filling first the rings and then the various bridges. 
The difference here is that there are additional and more complicated constraints to be considered, i.e. constraints 1-3\footnote{Note here that for P2P edges, only (the simple) constraint 4 needs to be regarded.}. 
One edge is added at a time, as with the P2P edges, but only if it does not violate any constraint. 
If the latter is not satisfied, we choose another provider or customer, depending on the violation, and try again.

The model is completed with the addition of the hangers. 
The rich get richer technique is also used here, keeping the upper and lower bounds in mind (constraint 4).  

In the following section we evaluate the performance of our generation proccess.

\section{Evaluation of our model}
\label{sec:Evaluation}
\setcounter{paragraph}{0}

\newtheorem{proposition}{\bf Proposition}

In this section we will present some results from the evaluation of our graph generator.  

As mentioned above, some of the main metrics that are used for describing properties of the graphs are the node degree distribution, the max/min/average node degree, the (effective) diameter of the graph as well as the clustering coefficient \cite{Mahadevan06}.

Using existing tools \cite{CAIDA}, we calculate a set of graph metrics for both the actual snapshot of the Internet \cite{He}, as well as the graph generated from the process described in the previous section.  
Table \ref{tab:ev_metrics} presents the results obtained.

Both graphs, the Internet snapshot \cite{He} and the graph obtained from our generator, are directed. 
The direction of each edge is based on the definition of the ASs relationships. 
P2P edges can be considered bidirectional, while the direction of CP edges can be defined based on whether we view the edge as a citation from a customer to its provider or a service offer from a provider to a customer. 
Nevertheless, we would like to stress out that the results presented here refer to the corresponding undirected graph \cite{CAIDA}.    
Consequently, since the metrics do not consider direction, we only obtain structural information about the graphs. 
The direction of an edge represents the type of the corresponding link (P2P/CP). 
Even though the undirected metrics can not be used for evaluating the quality of the reproduction of this information, the very nature of our strategy guarantees an accurate generation of the directed edges, at least in terms of quantity and distribution.

\begin{table}

\begin{center}
		\centerline{\begin{tabular}{|l|l|l|}
			\hline
			Metrics & Our graph & Graph at \cite{He} \\ \hline
			\# nodes & $19422$ & $19936$ \\ \hline
			\# edges & $59806$ & $59508$ \\ \hline
			Average node degree & $6.1$ & $5.9$ \\ \hline
			Max node degree & $2521$ & $2430$ \\ \hline
			Mutual Information Ratio & $0.96$ & $0.95$ \\ \hline
			Maximum local clustering & $1$ & $1$ \\ \hline
			Clustering coefficient & $0.052$ & $0.061$\\ 
			\hline
		\end{tabular}}
	\end{center}
\caption{A comparison table with the metrics computed at both graphs.}
\label{tab:ev_metrics}
\end{table}

\begin{figure*}[ht]
\centering
\subfigure[The node degree distribution for the given snapshot of the Internet.] 
{
    \label{fig:sub:a}
    \includegraphics[width=6.5cm]{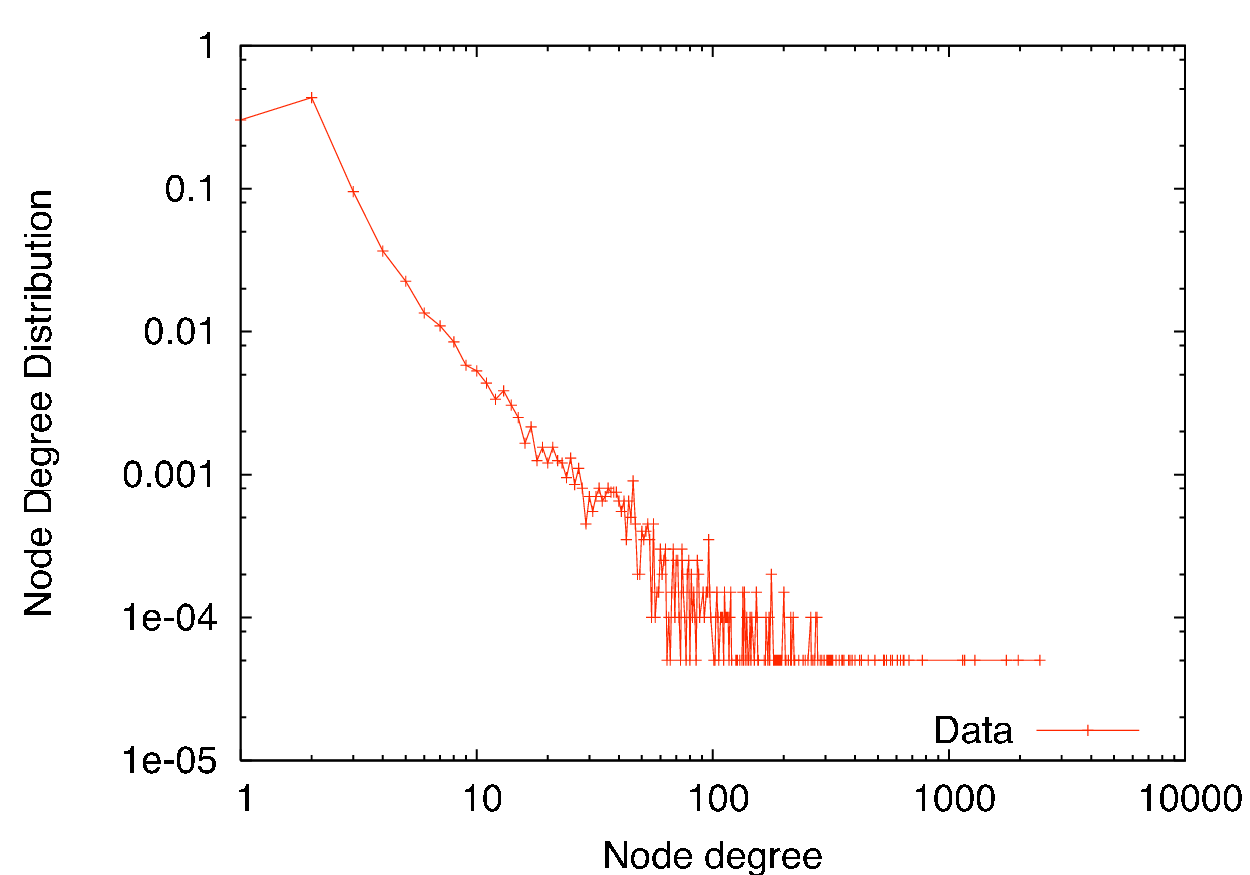}
}
\hspace{1cm}
\subfigure[The node degree distribution for the graph that we generated.] 
{
    \label{fig:sub:b}
    \includegraphics[width=6.5cm]{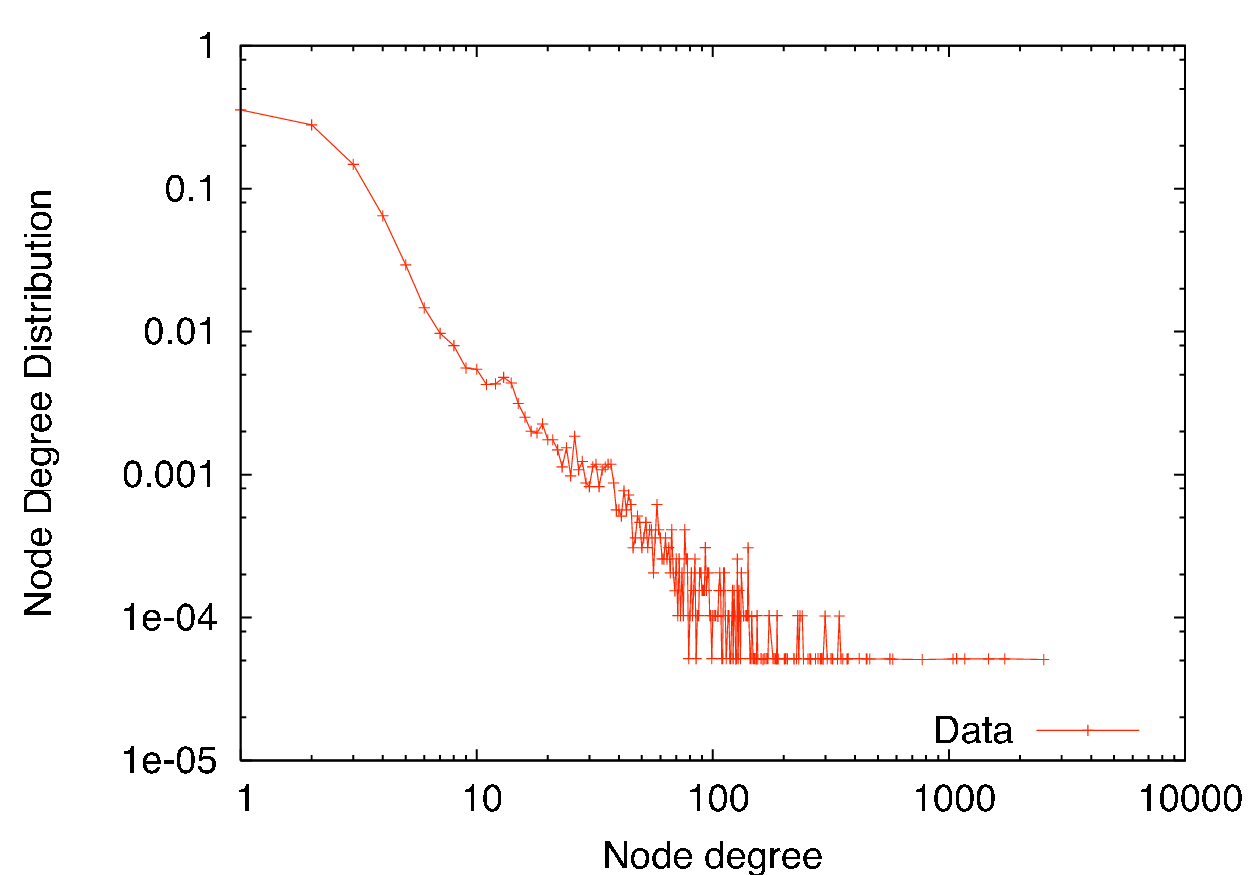}
}
\caption{The node degree distribution for both the snapshot of the Internet that we had and for the graph we generated.}
\label{fig:sub-deg} 
\end{figure*}

\begin{figure*}[ht]
\centering
\subfigure[The node degree CCDF for the given snapshot of the Internet.] 
{
    \label{fig:sub:a}
    \includegraphics[width=6.5cm]{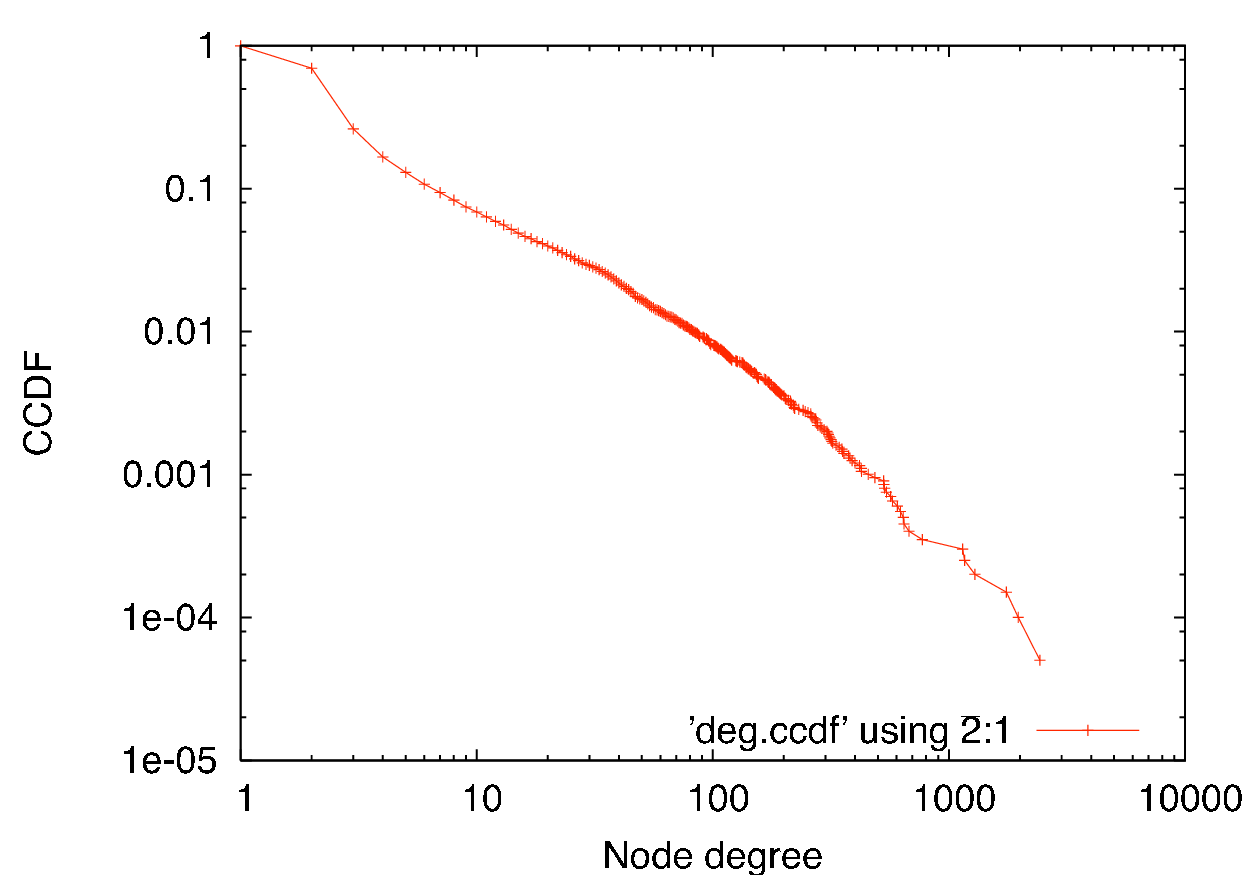}
}
\hspace{1cm}
\subfigure[The node degree CCDF for the graph that we generated.] 
{
    \label{fig:sub:b}
    \includegraphics[width=6.5cm]{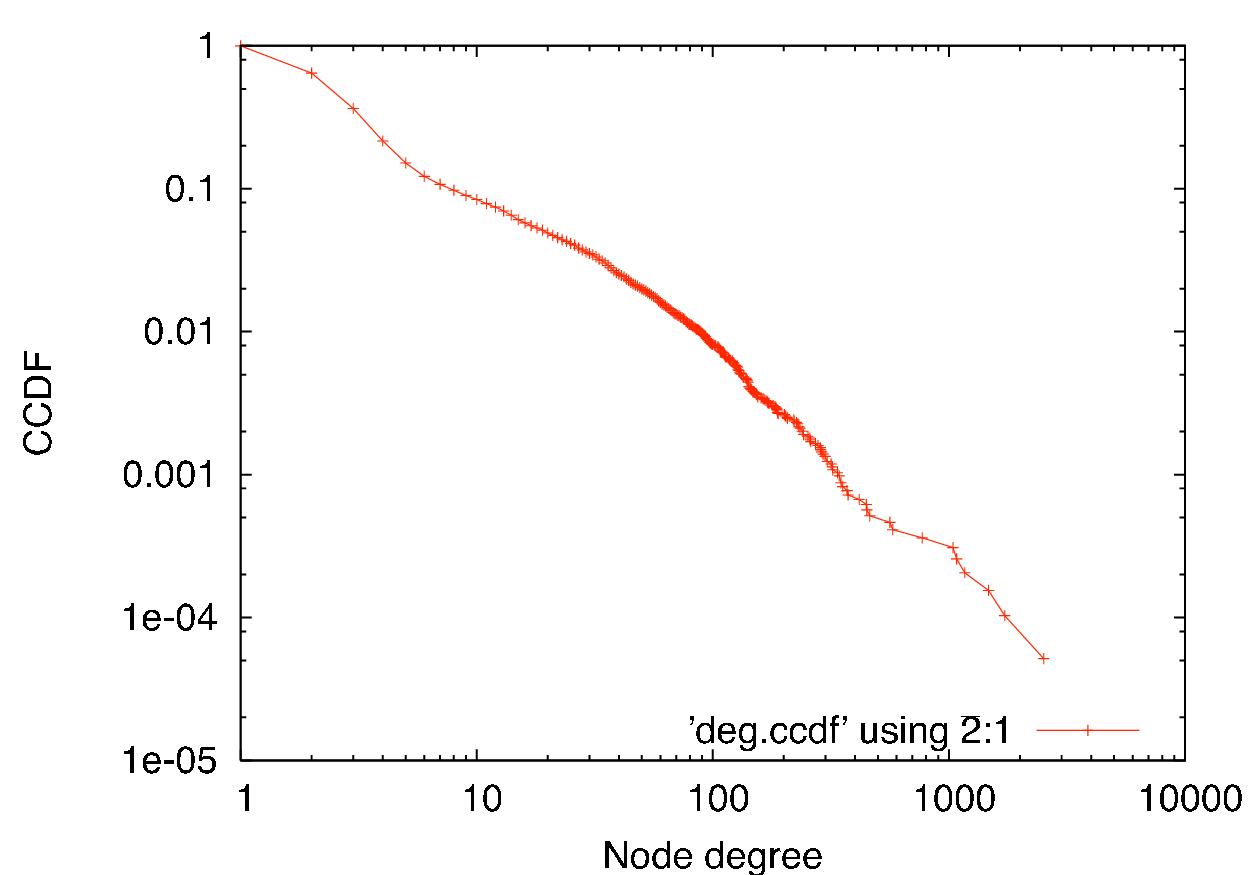}
}
\caption{The node degree CCDF for both the snapshot of the Internet that we had and for the graph we generated.}
\label{fig:sub-ccdf} 
\end{figure*}

Our evaluations presented in Table \ref{tab:ev_metrics}, reveal that the calculated metrics for the two graphs have very similar values.  
As a result, the graph that our generator created exhibits the same structural properties with the internet snapshot we provided to it as input.  
Furthermore, we can show the following proposition. 

\begin{proposition}\label{pr:diameter}
The effective diameter of our graph is the same with the one reported by empirical studies of the Internet AS topology \cite{Siganos03}.
\end{proposition}

\begin{proof}
As mentioned in \cite{Siganos06} there is an upper bound on the distance between any node $u$ of layer $k_{u}$ and any node $w$ of layer $k_{w}$.  This upper bound is: 

\begin{equation}
d(u,w)=k_{u}+k_{w}+1
\label{eq:upper_bound}
\end{equation}

Furthermore, we have seen that almost 90\% of the nodes belong to layer 0 through 2. 
Thus, 90\% of the nodes are within 5 hops away from each other.  
Recall that the above is the definition for the {\em effective diameter}.   
The graphs generated from our generator are therefore guaranteed to exhibit an effective diameter of 5\footnote{This fact has also been verified with our evaluations.}, which is also the effective diameter observed for the Internet AS level topology \cite{Siganos03}.
\end{proof}

Finally, Figures \ref{fig:sub-deg} and \ref{fig:sub-ccdf} present our results for the node degree distribution and the node degree CCDF - Complementary Cumulative Distribution Function - for both our graph and the original AS topology used. 
The node degree distribution is the plot of the number of nodes with a specific degree versus their degree (i.e. the histogram of the node degree), while the node degree CCDF represents the probability that the degree of a randomly picked node is greater than a certain value.  
What we can observe here is that both metrics provide a close match again, in the sense that the general trend in both graphs is the same. 
This further supports the fact that \textbf{{\em our graph generator accurately captures the structural properties of the Internet graph at the AS level topology}}.

\section{Conclusions - Future work}
\label{sec:conclusions}
\setcounter{paragraph}{0}

The main goal of our work is to provide a tool, able to reproduce with high accuracy the Internet topology at the AS level.  
Unlike previous efforts, we focus on a conceptual model for the Internet topology, which can effectively capture the various properties of the real AS level graph.  
To reiterate, the main contributions of our work are:
\begin{itemize}
\item We identified and studied the jellyfish structure using the largest available snapshot of the Internet.
\item We created (designed and implemented) a graph generator based on the conceptual model of the jellyfish for the Internet and the AS relations.
\item We evaluated our work, by using several metrics that can effectively capture the structure of a graph.
\end{itemize}

The outcome of our work is very promising. 
The graph that our generator produces exhibits all the important features of the conceptual model. 
Most graph generators up to now were focused on a single, specific metric, which was thought to capture most of the structural information of the input graph. 
In this work, we proposed a different, novel approach, in which the objective is to emulate a conceptual model as opposed to a specific metric.  
The novelty of our work can be attributed to the following two points:

\begin{enumerate}
\item Our generator identifies the jellyfish-like structure of the Internet \cite{Siganos06} and uses it to construct the backbone of the new graph.  
\item The creation of the new graph is driven by the noted relations between the ASs \cite{Huston99} (P2P and CP relations).
\end{enumerate}

Some interesting issues still to be considered can be summarized in the following:

\begin{itemize}
\item Up to this point we have only taken into account metrics for undirected graphs. 
It will be of great interest to use metrics for directed graphs in order to further evaluate and improve our generation strategy.
\item More sophisticated statistical methods can possibly be applied in order to further enhance the accuracy and performance of our generator.  
\item As mentioned in Section \ref{sec:approach}, graph shrinking, is of great importance for the research community (e.g. minimizing the simulation cost).  
In the future, we seek to examine the capabilities of our generator to form an accurate graph shrinking tool as well.
\end{itemize}

We opt to address the above issues in our future work, creating by this way a complete Internet AS level graph generator available to the research community.  
Finally, as a step further, we are interested into applying our generic conceptual approach on different types of graphs that follow different conceptual models.  
A higher level of abstraction will be thus imposed to our graph generator, increasing at the same time its possible applications.
\bibliographystyle{unsrt}
\bibliography{main}

\end{document}